\documentclass[twocolumn,numbers]{elsarticle}
\usepackage{fullpage}
\usepackage{amsthm,amssymb}

\journal{}

\newtheorem{theorem}{Theorem}

\newtheorem{proposition}{Proposition}

\newtheorem{lemma}{Lemma}

\usepackage{xcolor}
\usepackage[dvipsnames]{xcolor}

\newcommand{\cA}{\mathcal{A}}

\newcommand{\cR}{\mathcal{R}}
\newcommand{\cD}{\mathcal{D}}
\newcommand{\cP}{\mathcal{P}}
\newcommand{\duet}[2]{\overline{#1\,\, #2}}

\usepackage{tikz-cd}
\usetikzlibrary{fit}
\def\FigOneScale{0.4}
\def\FigTwoScale{0.5}
\def\FigThreeScale{0.425}
\def\FigFourScale{0.3}

\begin{document}

\begin{frontmatter}

\title{Efficient Reconstruction of Arboreal Networks}

\author[uea]{Katharina T. Huber} 
\author[cuny,amnh]{Katherine St.~John}

\affiliation[uea]{
    organization={School of Computing Sciences, University of East Anglia},
    city = {Norwich}, 
    country = {UK}.}

\affiliation[cuny]{
    organization={Department of Computer Science, Hunter College and Graduate Center, City University of New York},
    city = {New York}, 
    state = {NY}, 
    country = {USA.}}

\affiliation[amnh]{
    organization={Division of Invertebrate Zoology, American Museum of Natural History},
    city = {New York}, 
    state = {NY}, 
    country = {USA.}}

\end{frontmatter}

\section*{Abstract}
Arboreal networks are multi-rooted phylogenetic networks whose underlying graph is a tree.
We give an encoding 
of stack-free arboreal networks in terms of triplets and the novel concept of a duet.
This 
yields a polynomial time algorithm to construct these networks from complete triplet and duet systems.
The classification results show correctness and lead to a natural metric on these multi-rooted networks.

\noindent
{\em Keywords:  arboreal phylogenetic network, triplet \& duet systems, reconstruction algorithms, metrics.}

\section{Introduction}

{An improved understanding of the complex processes that drive molecular evolution has led to extensive interest in phylogenetic networks to represent them. Recently introduced classes of multi-rooted (phylogenetic) networks include forest-based networks \cite{huber2022forest} and arboreal networks \cite{HMS24}.  Both extend the notion of the popular tree-based networks \cite{FS15} and may be useful for studying lateral gene transfer between bacteria living in ecologically distinct niches such as the hand and the human gut \cite{jeong2019}. Both are multi-rooted directed acyclic graphs whose leaf set is a fixed set of taxa, and the study of their combinatorial structure has revealed intriguing links
with, for example, Ptolemaic graphs \cite{HMS24}. 
The computational complexity of recognizing forest-based networks was studied in \cite{SPTMH19} under various scenarios. 
Their class 
encompasses the class of Overlaid Species Forests which have been used to shed light into 
introgression in Heliconius butterflies in \cite{huber2025network}.} 

We focus on arboreal networks, an intriguing class of multi-rooted networks whose underlying structure is a tree.
We give a new combinatorial classification for 
arboreal networks 
which yields a natural metric on the duets (rooted 2-leaf subtrees) and triplets (rooted 3-leaf {binary} subtrees) compatible with arboreal networks.  
Using this characterization, we give a polynomial-time algorithm to reconstruct an arboreal network, if it exists, from the complete triplet and duet systems induced by the network.
Allowing multiple roots in a network adds power to capture biological relationships, but, as we shall see, it also adds complexity.  
Our approach relies on the {\sc Build} algorithm \cite{aho81}.  While the {\sc Build} algorithm cannot directly reconstruct multi-rooted networks, 
a pre-processing step, based on duets, and additional clustering yields an efficient algorithm that matches {the} running time bounds of the {\sc Build} algorithm.

\section{Background}

We follow the notation of \cite{huber2022forest,huber2025network}, unless otherwise noted.
Throughout, we assume that the leaf-label set is a finite set $X$ with $|X| \ge 3$. 
Let $G$ be a directed simple graph with vertex set $V(G)$ and arc (directed edge) set $E(G)$.
For any multi-rooted directed acyclic graph ({\bf mDAG}) $G$, we call a vertex of $G$ 
with in-degree 0 and out-degree at least 2 a {\bf root} of $G$ and a vertex
with in-degree 1 and out-degree zero {a {\bf leaf}} of $G$.  We denote the set of leaves of $G$ by $L(G)$ and call every vertex of $G$ that is neither a root nor a leaf an {\bf interior vertex} of $G$.  
If all of the outward arcs of an {interior} vertex of $G$ are adjacent to leaves, we call that set of leaves a {\bf generalized cherry}. If a generalized cherry has size two, then we call it a {\bf cherry}.
We say that a vertex $v$ of $G$ is {\bf above} a vertex $u$ of $G$ if there exists a root $\rho$ of $G$ such that {a} path from $\rho$ to $u$ crosses $v$, and we say {that} $u$ is {\bf strictly above} $v$ if $u$ is above $v$ and $u\not=v$. We say that $u$ is {\bf below} $v$ if $v$ is above $u$. We let $U(G)$ denote the graph that results from $G$ by ignoring the directions of its arcs.
We call the graph, $U(G)^-$, obtained from  $U(G)$ by suppressing all vertices of degree 2, the {\bf underlying graph} of $G$ (see Fig.~\ref{fig:not-encoded}(e,f)).
An unrooted tree $T$ 
is an {\bf unrooted phylogenetic tree} 
if $T$ contains no vertices of degree 2.

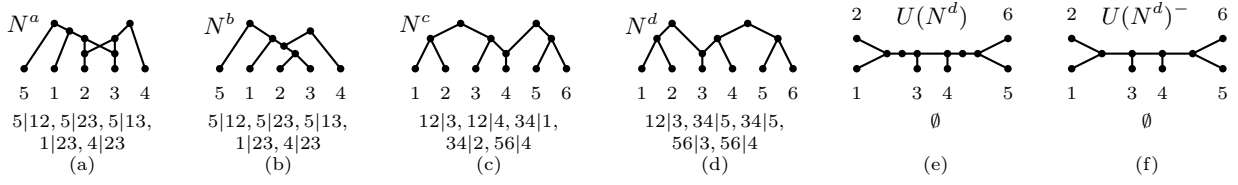
\begin{figure*}[t]
\begin{center}
{\scriptsize
\begin{tabular}{cccccc}
    \begin{tikzpicture}[thick,scale=\FigOneScale]
    \node[align=left] at (1,2.5) {\small $N^a$};
    \node[outer sep=2.5pt, fill=black,circle,inner sep=1pt, label=below: {$5$} ] at (1,1){};
    \node[outer sep=2.5pt, fill=black,circle,inner sep=1pt, label=below: {$1$} ] at (2,1){};
    \node[outer sep=2.5pt, fill=black,circle,inner sep=1pt, label=below: {$2$} ] at (3,1){};
    \node[outer sep=2.5pt, fill=black,circle,inner sep=1pt, label=below: {$3$} ] at (4,1){};
    \node[outer sep=2.5pt, fill=black,circle,inner sep=1pt, label=below: {$4$} ] at (5,1){};
    \node[fill=black,circle,inner sep=1pt] at (3,1.5){};
    \node[fill=black,circle,inner sep=1pt] at (4,1.5){};
    \node[fill=black,circle,inner sep=1pt] at (3,2){};
    \node[fill=black,circle,inner sep=1pt] at (4,2){};
    \node[fill=black,circle,inner sep=1pt] at (2.5,2.25){};
    \node[fill=black,circle,inner sep=1pt] at (2,2.5){};
    \node[fill=black,circle,inner sep=1pt] at (4.5,2.5){};
    \draw(3,1)--(3,2);
    \draw(4,1)--(4,2);
    \draw(3,1.5)--(4,2);
    \draw(4,1.5)--(3,2);
    \draw(2.5,2.25)--(3,2);
    \draw(2,2.5)--(2.5,2.25);
    \draw(2,2.5)--(1,1);
    \draw(2,1)--(2.5,2.25);
    \draw(5,1)--(4.5,2.5);
    \draw(4,2)--(4.5,2.5);
\end{tikzpicture}
    &\begin{tikzpicture}[thick,scale=\FigOneScale]
    \node[align=left] at (1,2.5) {\small $N^b$};
    \node[outer sep=2.5pt, fill=black,circle,inner sep=1pt, label=below: {$5$} ] at (1,1){};
    \node[outer sep=2.5pt, fill=black,circle,inner sep=1pt, label=below: {$1$} ] at (2,1){};
    \node[outer sep=2.5pt, fill=black,circle,inner sep=1pt, label=below: {$2$} ] at (3,1){};
    \node[outer sep=2.5pt, fill=black,circle,inner sep=1pt, label=below: {$3$} ] at (4,1){};
    \node[outer sep=2.5pt, fill=black,circle,inner sep=1pt, label=below: {$4$} ] at (5,1){};
    \node[fill=black,circle,inner sep=1pt] at (3.5,1.5){};
    \node[fill=black,circle,inner sep=1pt] at (3.125,1.75){};
    \node[fill=black,circle,inner sep=1pt] at (2.75,2){};
    \node[fill=black,circle,inner sep=1pt] at (2,2.5){};
    \node[fill=black,circle,inner sep=1pt] at (4,2.25){};
    \draw(3,1)--(3.5,1.5);
    \draw(4,1)--(3.5,1.5);
    \draw(2,2.5)--(2.75,2);
    \draw(2,2.5)--(1,1);
    \draw(2,1)--(2.75,2);
    \draw(3.5,1.5)--(2.75,2);
    \draw(5,1)--(4,2.25);
    \draw(3.125,1.75)--(4,2.25);
\end{tikzpicture}
    &\begin{tikzpicture}[thick,scale=\FigOneScale]
    \node[align=left] at (1,2.5) {\small $N^c$};
    \node[outer sep=2.5pt, fill=black,circle,inner sep=1pt, label=below: {$1$} ] at (1,1){};
    \node[outer sep=2.5pt, fill=black,circle,inner sep=1pt, label=below: {$2$} ] at (2,1){};
    \node[outer sep=2.5pt, fill=black,circle,inner sep=1pt, label=below: {$3$} ] at (3,1){};
    \node[outer sep=2.5pt, fill=black,circle,inner sep=1pt, label=below: {$4$} ] at (4,1){};
    \node[outer sep=2.5pt, fill=black,circle,inner sep=1pt, label=below: {$5$} ] at (5,1){};
    \node[outer sep=2.5pt, fill=black,circle,inner sep=1pt, label=below: {$6$} ] at (6,1){};
    \node[fill=black,circle,inner sep=1pt] at (1.5,2){};
    \node[fill=black,circle,inner sep=1pt] at (3.5,2){};
    \node[fill=black,circle,inner sep=1pt] at (4,1.5){};
    \node[fill=black,circle,inner sep=1pt] at (5.5,2){};
    \node[fill=black,circle,inner sep=1pt] at (5,2.5){};
    \node[fill=black,circle,inner sep=1pt] at (2.5,2.5){};
    \draw(2.5,2.5)--(1.5,2);
        \draw(1.5,2)--(1,1);
        \draw(1.5,2)--(2,1);
    \draw(2.5,2.5)--(3.5,2);
        \draw(3.5,2)--(3,1);
        \draw(3.5,2)--(4,1.5);
        \draw(4,1)--(4,1.5);
    \draw(5,2.5)--(4,1.5);
    \draw(5,2.5)--(5.5,2);
        \draw(5.5,2)--(5,1);
        \draw(5.5,2)--(6,1);
\end{tikzpicture}
    &\begin{tikzpicture}[thick,scale=\FigOneScale]
    \node[align=left] at (1,2.5) {\small $N^d$};
    \node[outer sep=2.5pt, fill=black,circle,inner sep=1pt, label=below: {$1$} ] at (1,1){};
    \node[outer sep=2.5pt, fill=black,circle,inner sep=1pt, label=below: {$2$} ] at (2,1){};
    \node[outer sep=2.5pt, fill=black,circle,inner sep=1pt, label=below: {$3$} ] at (3,1){};
    \node[outer sep=2.5pt, fill=black,circle,inner sep=1pt, label=below: {$4$} ] at (4,1){};
    \node[outer sep=2.5pt, fill=black,circle,inner sep=1pt, label=below: {$5$} ] at (5,1){};
    \node[outer sep=2.5pt, fill=black,circle,inner sep=1pt, label=below: {$6$} ] at (6,1){};
    \node[fill=black,circle,inner sep=1pt] at (1.5,2){};
    \node[fill=black,circle,inner sep=1pt] at (3,1.5){};
    \node[fill=black,circle,inner sep=1pt] at (3.5,2){};
    \node[fill=black,circle,inner sep=1pt] at (5.5,2){};
    \node[fill=black,circle,inner sep=1pt] at (4.5,2.5){};
    \node[fill=black,circle,inner sep=1pt] at (2,2.5){};
    \draw(2,2.5)--(1.5,2);
        \draw(1.5,2)--(1,1);
        \draw(1.5,2)--(2,1);
    \draw(2,2.5)--(3,1.5);
        \draw(3,1.5)--(3,1);
    \draw(4,1)--(3.5,2);
    \draw(3.5,2)--(3,1.5);
    \draw(3.5,2)--(4.5,2.5);
    \draw(4.5,2.5)--(5.5,2);
        \draw(5.5,2)--(5,1);
        \draw(5.5,2)--(6,1);
\end{tikzpicture}
    &\begin{tikzpicture}[thick,scale=\FigOneScale]
    \node[align=left] at (3.5,2.75) {\small $U(N^d)$};
    \node[outer sep=2.5pt, fill=black,circle,inner sep=1pt, label=below: {$1$} ] at (1,1){};
    \node[outer sep=2.5pt, fill=black,circle,inner sep=1pt, label=above: {$2$} ] at (1,2){};
    \node[outer sep=2.5pt, fill=black,circle,inner sep=1pt, label=below: {$3$} ] at (3,1){};
    \node[outer sep=2.5pt, fill=black,circle,inner sep=1pt, label=below: {$4$} ] at (4,1){};
    \node[outer sep=2.5pt, fill=black,circle,inner sep=1pt, label=below: {$5$} ] at (6,1){};
    \node[outer sep=2.5pt, fill=black,circle,inner sep=1pt, label=above: {$6$} ] at (6,2){};
    \node[fill=black,circle,inner sep=1pt] at (2,1.5){};
    \node[fill=black,circle,inner sep=1pt] at (3,1.5){};
    \node[fill=black,circle,inner sep=1pt] at (4,1.5){};
    \node[fill=black,circle,inner sep=1pt] at (5,1.5){};
    \node[fill=black,circle,inner sep=1pt] at (4.5,1.5){};
    \node[fill=black,circle,inner sep=1pt] at (2.5,1.5){};
    \draw(1,1)--(2,1.5);
    \draw(1,2)--(2,1.5);
    \draw(2,1.5)--(5,1.5);
    \draw(5,1.5)--(6,1);
    \draw(5,1.5)--(6,2);
    \draw(3,1)--(3,1.5);
    \draw(4,1)--(4,1.5);
\end{tikzpicture}
    &\begin{tikzpicture}[thick,scale=\FigOneScale]
    \node[align=left] at (3.45,2.75) {\small $U(N^d)^-$};
    \node[outer sep=2.5pt, fill=black,circle,inner sep=1pt, label=below: {$1$} ] at (1,1){};
    \node[outer sep=2.5pt, fill=black,circle,inner sep=1pt, label=above: {$2$} ] at (1,2){};
    \node[outer sep=2.5pt, fill=black,circle,inner sep=1pt, label=below: {$3$} ] at (3,1){};
    \node[outer sep=2.5pt, fill=black,circle,inner sep=1pt, label=below: {$4$} ] at (4,1){};
    \node[outer sep=2.5pt, fill=black,circle,inner sep=1pt, label=below: {$5$} ] at (6,1){};
    \node[outer sep=2.5pt, fill=black,circle,inner sep=1pt, label=above: {$6$} ] at (6,2){};
    \node[fill=black,circle,inner sep=1pt] at (2,1.5){};
    \node[fill=black,circle,inner sep=1pt] at (3,1.5){};
    \node[fill=black,circle,inner sep=1pt] at (4,1.5){};
    \node[fill=black,circle,inner sep=1pt] at (5,1.5){};
    \draw(1,1)--(2,1.5);
    \draw(1,2)--(2,1.5);
    \draw(2,1.5)--(5,1.5);
    \draw(5,1.5)--(6,1);
    \draw(5,1.5)--(6,2);
    \draw(3,1)--(3,1.5);
    \draw(4,1)--(4,1.5);
\end{tikzpicture}\\
    $5|12, 5|23,5|13,$&$5|12, 5|23,5|13,$&$12|3,12|4,34|1,$&$12|3,34|5,34|5,$&$\emptyset$&$\emptyset$\\
    $1|23, 4|23$&$1|23, 4|23$&$34|2,56|4$&$56|3,56|4$&&\\
    (a)&(b)&(c)&(d)&(e)&(f)\\
\end{tabular}
}
\end{center}
\caption{Encoding Networks and their induced triplet systems: The 2-network $N^a$ 
is not encoded by its induced triplets since $\cR(N^a)=\{5|12, 5|23,5|13, 1|23, 4|23\}=\cR(N^b)$ holds for the $2$-network $N^b$. $N^a$ is not arboreal, but $N^b$ is.
Two arboreal networks $N^c$ and $N^d$
that have the same underlying {phylogenetic} tree given in (f) and obtained from (e).
}
\label{fig:not-encoded}
\label{fig:aboreal-from-tree}
\end{figure*}

An {\bf $m$-network} $N$ (on $X$) is an mDAG with $m\geq 1$ roots such that $U(N)$ is connected, $L(N)=X$,
a root has in-degree zero and out-degree 2, there are no vertices of in-degree 1 and out-degree 1, and an interior vertex has either out-degree 1 and in-degree at least 2 ({\bf hybrid}) or in-degree 1 and out-degree at least 2 ({\bf tree vertex}). We call  an ordered pair $(a,b)$ with $a\ne b \in X$ a {\bf reticulated cherry} of $N$ if $b$ is the child of a hybrid $h$ of $N$ and $h$ and $a$ share a parent. For example, $(3,4)$ is a reticulated cherry in the network in Fig.~\ref{fig:not-encoded}(c).
Note that if $m=1$ and the out-degree of the root is allowed to be 2 or more, then $N$ is also called a {\bf rooted phylogenetic tree (on $X$)}. Two $m$-networks $N$ and $N'$ on $X$ are {\bf isomorphic}, {denoted} by $N\simeq N'$, if they are isomorphic as mDAGs and this isomorphism is the identity on $X$. 

A $m$-network $N$ is called {\bf arboreal} if $U(N)^-$ is an unrooted phylogenetic tree. 
For $v\in V(N)$, we define $L(v)$ to be the set of all leaves of $N$ that are below $v$. 
We define $N_v$ to be the spanning tree of $N$ with leaf set $L(v)$.  
There can be multiple arboreal networks for the same underlying phylogenetic tree (see Fig.~\ref{fig:not-encoded}).
We call a rooted phylogenetic tree on 
$X=\{x,y,z\}$ with cherry $\{x,y\}$  a {\bf (rooted) triplet 
(on $X$)} and denote it by $xy|z$ (or alternatively, by
$z|xy$) where, the order of $x$ and $y$ does not matter. We refer to a set $\cR$ of triplets as a {\bf triplet system (on $L(\cR):=\bigcup_{t\in \cR}L(t)$)}. These arise naturally in the context of arboreal networks. Let $N$ be an arboreal network on $X$ and let $x,y,z\in X$ denote three pairwise distinct elements. Then we say that the triplet $xy|z$ is {\bf induced} by $N$ if 
there  exists a root $r$ of $N$ such that $x,y,z\in L(r)$ and $z$ is not below the least common ancestor of $x$ or $y$. For example, the triplet $56|4$ is induced by the arboreal network 
in Fig.~\ref{fig:aboreal-from-tree}(c).
The set $\cR(N)$ of all triplets induced by an arboreal network, $N$, can be empty as in  Fig.~\ref{fig:mountain-range}(e).

\begin{figure*}[t]
\begin{center}
{\scriptsize
\begin{tabular}{r c c c c c c}
    &\begin{tikzpicture}[thick,scale=\FigTwoScale]
    \node[align=left] at (1,2.5) {\small $N^a$};
    \node[outer sep=2.5pt, fill=black,circle,inner sep=1pt, label=below: {$1$} ] at (1,1){};
    \node[outer sep=2.5pt, fill=black,circle,inner sep=1pt, label=below: {$2$} ] at (2,1){};
    \node[outer sep=2.5pt, fill=black,circle,inner sep=1pt, label=below: {$3$} ] at (3,1){};
    \node[outer sep=2.5pt, fill=black,circle,inner sep=1pt, label=below: {$4$} ] at (4,1){};
    \node[fill=black,circle,inner sep=1pt] at (1.5,1.75){};
    \node[fill=black,circle,inner sep=1pt] at (2,2.5){};
    \draw(1,1)--(1.5,1.75);
    \draw(2,1)--(1.5,1.75);
    \draw(1.5,1.75)--(2,2.5);
    \draw(3,1)--(2,2.5);
    \draw(4,1)--(2,2.5);
\end{tikzpicture}
    &\begin{tikzpicture}[thick,scale=\FigTwoScale]
    \node[align=left] at (1,2.5) {\small $N^b$};
    \node[outer sep=2.5pt, fill=black,circle,inner sep=1pt, label=below: {$1$} ] at (2,1){};
    \node[outer sep=2.5pt, fill=black,circle,inner sep=1pt, label=below: {$2$} ] at (3,1){};
    \node[outer sep=2.5pt, fill=black,circle,inner sep=1pt, label=below: {$3$} ] at (1,1){};
    \node[outer sep=2.5pt, fill=black,circle,inner sep=1pt, label=below: {$4$} ] at (4,1){};
    \node[fill=black,circle,inner sep=1pt] at (2.5,1.5){};
    \node[fill=black,circle,inner sep=1pt] at (2.5,2){};
    \node[fill=black,circle,inner sep=1pt] at (2,2.5){};
    \node[fill=black,circle,inner sep=1pt] at (3,2.5){};
    \draw(1,1)--(2,2.5);
    \draw(2,1)--(2.5,1.5);
    \draw(3,1)--(2.5,1.5);
    \draw(4,1)--(3,2.5);
    \draw(2.5,1.5)--(2.5,2);
    \draw(2.5,2)--(2,2.5);
    \draw(2.5,2)--(3,2.5);
\end{tikzpicture}
    &\begin{tikzpicture}[thick,scale=\FigTwoScale]
    \node[align=left] at (1,2.5) {\small $N^c$};
    \node[outer sep=2.5pt, fill=black,circle,inner sep=1pt, label=below: {$1$} ] at (2,1){};
    \node[outer sep=2.5pt, fill=black,circle,inner sep=1pt, label=below: {$2$} ] at (3,1){};
    \node[outer sep=2.5pt, fill=black,circle,inner sep=1pt, label=below: {$3$} ] at (1,1){};
    \node[outer sep=2.5pt, fill=black,circle,inner sep=1pt, label=below: {$4$} ] at (4,1){};
    \node[fill=black,circle,inner sep=1pt] at (2.5,1.75){};
    \node[fill=black,circle,inner sep=1pt] at (2,2.5){};
    \node[fill=black,circle,inner sep=1pt] at (3.5,1.75){};
    \node[fill=black,circle,inner sep=1pt] at (2.875,1.25){};
    \draw(2,1)--(2.5,1.75);
    \draw(3,1)--(2.5,1.75);
    \draw(1,1)--(2,2.5);
    \draw(2.5,1.75)--(2,2.5);
    \draw(4,1)--(3.5,1.75);
    \draw(2.875,1.25)--(3.5,1.75);
\end{tikzpicture}
    &\begin{tikzpicture}[thick,scale=\FigTwoScale]
    \node[align=left] at (1,2.5) {\small $N^d$};
    \node[outer sep=2.5pt, fill=black,circle,inner sep=1pt, label=below: {$2$} ] at (2,1){};
    \node[outer sep=2.5pt, fill=black,circle,inner sep=1pt, label=below: {$1$} ] at (3,1){};
    \node[outer sep=2.5pt, fill=black,circle,inner sep=1pt, label=below: {$3$} ] at (1,1){};
    \node[outer sep=2.5pt, fill=black,circle,inner sep=1pt, label=below: {$4$} ] at (4,1){};
    \node[fill=black,circle,inner sep=1pt] at (2.5,1.75){};
    \node[fill=black,circle,inner sep=1pt] at (2,2.5){};
    \node[fill=black,circle,inner sep=1pt] at (3.5,1.75){};
    \node[fill=black,circle,inner sep=1pt] at (2.875,1.25){};
    \draw(2,1)--(2.5,1.75);
    \draw(3,1)--(2.5,1.75);
    \draw(1,1)--(2,2.5);
    \draw(2.5,1.75)--(2,2.5);
    \draw(4,1)--(3.5,1.75);
    \draw(2.875,1.25)--(3.5,1.75);
\end{tikzpicture}
    &\begin{tikzpicture}[thick,scale=\FigTwoScale]
    \node[align=left] at (0.9,2.5) {\small $N^e$};
    \node[outer sep=2.5pt, fill=black,circle,inner sep=1pt, label=below: {$1$} ] at (2.5,1){};
    \node[outer sep=2.5pt, fill=black,circle,inner sep=1pt, label=below: {$2$} ] at (1,1){};
    \node[outer sep=2.5pt, fill=black,circle,inner sep=1pt, label=below: {$3$} ] at (4,2){};
    \node[outer sep=2.5pt, fill=black,circle,inner sep=1pt, label=below: {$4$} ] at (3.5,1){};
    \node[fill=black,circle,inner sep=1pt] at (2.5,1.5){};
    \node[fill=black,circle,inner sep=1pt] at (3,2.5){};
    \node[fill=black,circle,inner sep=1pt] at (3,1.75){};
    \node[fill=black,circle,inner sep=1pt] at (2,2){};
    \node[fill=black,circle,inner sep=1pt] at (1.5,2.5){};
    \draw(1,1)--(1.5,2.5);
    \draw(2.5,1)--(2.5,1.5);
    \draw(2.5,1.5)--(1.5,2.5);
    \draw(3,2.5)--(4,2);
    \draw(3,2.5)--(2,2);
    \draw(3.5,1)--(3,1.75);
    \draw(2.5,1.5)--(3,1.75);
\end{tikzpicture}
    &\begin{tikzpicture}[thick,scale=\FigTwoScale]
    \node[align=left] at (0.9,2.5) {\small $N^f$};
    \node[outer sep=2.5pt, fill=black,circle,inner sep=1pt, label=below: {$1$} ] at (2.5,1){};
    \node[outer sep=2.5pt, fill=black,circle,inner sep=1pt, label=below: {$2$} ] at (1,1){};
    \node[outer sep=2.5pt, fill=black,circle,inner sep=1pt, label=below: {$3$} ] at (4,2){};
    \node[outer sep=2.5pt, fill=black,circle,inner sep=1pt, label=below: {$4$} ] at (3.5,1){};
    \node[fill=black,circle,inner sep=1pt] at (2.5,1.5){};
    \node[fill=black,circle,inner sep=1pt] at (3,2.5){};
    \node[fill=black,circle,inner sep=1pt] at (3,1.75){};
    \node[fill=black,circle,inner sep=1pt] at (1.5,2.5){};
    \draw(1,1)--(1.5,2.5);
    \draw(2.5,1)--(2.5,1.5);
    \draw(2.5,1.5)--(1.5,2.5);
    \draw(3,2.5)--(4,2);
    \draw(3,2.5)--(2.5,1.5);
    \draw(3.5,1)--(3,1.75);
    \draw(2.5,1.5)--(3,1.75);
\end{tikzpicture}\\
    duets:&$\emptyset$&$\emptyset$&$\duet{2}{4}$&$\duet{1}{4}$&$\duet{1}{2},\duet{1}{3},\duet{1}{4}$&$\duet{1}{2},\duet{1}{3},\duet{1}{4}$\\
    triplets:& $12|3,12|4$&$12|3,12|4$&$12|3$&$12|3$&$\emptyset$&$\emptyset$\\
    &(a)&(b)&(c)&(d)&(e)&(f)\\
\end{tabular}
}
\end{center}

    \caption{Networks on $\{1,2,3,4\}$: 
    (a) An mDAG, $N^a$, that is a rooted phylogenetic tree, but not a $1$-network, since the root has out-degree 3. (b) There are no triplets or duets that imply a relationship between leaves $3$ and $4$ {in $N^b$}. {Thus,} the only arboreal network that induces $\{12|3,12|4\}=\mathcal R(N^a)$ is the $2$-network $N^b$. {This implies that $N^b$ is encoded  by $\cR(N^b)$ within the class $\cA(X)$ and illustrates the out-degree 2 requirement of a root of an $m$-network.} 
    (c) and (d) Arboreal  networks $N^c$ and $N^d$ with $\cR(N^d)=\mathcal R(N^c)\not=\emptyset\not=\mathcal D(N^c)$. 
    (e) An arboreal network $N^e$ that is not stack-free. (f) A stack-free arboreal network $N^f$ such that $\mathcal R(N^e)\cup\mathcal D(N^e)=\mathcal R(N^f)\cup\mathcal D(N^f)$. {Duet-triplet} distance (see Sec.~\ref{sec:metric}) examples:  $d(N^b,N^c) = 2$ and $d(N^e,N^f) = 0$ since duets and triplets cannot account for stacked hybrids.
    }
    \label{fig:mountain-range}
    \label{fig:stackfree-necessary}
    \label{fig:absense-of-evidence}
\end{figure*}
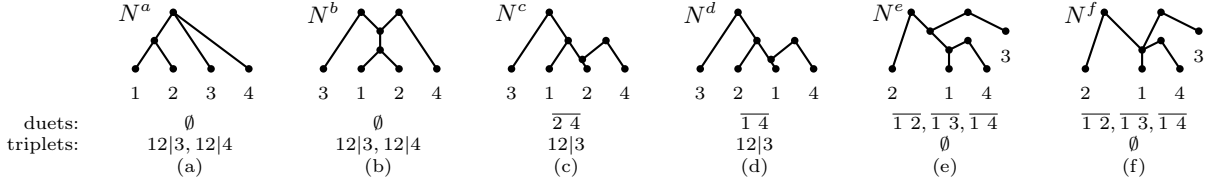

\section{Combinatorial Classifications}
\label{sec:characterization}

A rooted phylogenetic tree $T$  on $X$ is {\bf encoded} by
$\cR(T)$ {(within the class $\mathcal T$ of rooted phylogenetic trees)} if any other tree $T'$ in $\mathcal T$ for which $\mathcal R(T)=\mathcal R(T')$ is isomorphic with $T$ \cite{bryant,SS03}.  {Extending the notion of an encoding to 
$m$-networks  by replacing ``rooted phylogenetic tree" by ``$m$-network" in the definition of an encoding, then  even if  
a 1-network $N$ is binary and no two cycles in $U(N)^-$ share a vertex it is not, in general, encoded by $\cR(N)$ \cite{GH12}.}
Interestingly, such 1-networks are encoded by their induced set of trinets \cite{HM13} {where an encoding in terms of such structures is defined analogously to the triplet system case. Without going into details, a trinet} can be thought of as a natural generalization of a triplet to a $1$-network on 3 leaves (see \cite{OWvIM16,vIM14} for more on encodings of $1$-networks by trinets). Since trinets are, in particular, 1-networks it follows that they cannot directly be used as a starting point for finding an ``encoding" for general arboreal networks.

{Even if an $m$-network $N$ is arboreal it is not, in general, encoded by $\cR(N)$.} The two arboreal networks $N^c$ and $N^d$ in Fig.~\ref{fig:stackfree-necessary} induce the same triplet system but are not isomorphic. 
To distinguish these networks, we define a 
combinatorial structure on pairs of leaves that is induced by an {$m$-network}: we call an unordered pair of two distinct elements  $x,y\in X$ a {\bf duet} and denote it by $\duet{x}{y}$ or, alternatively, $\duet{y}{x}$.  Moreover, we refer to  a set $\cD$ of duets as a {\bf duet system (on $L(\cD):=\bigcup_{\duet{a}{b}\in \cD} \{a,b\}$).} Let $N$ be an {$m$-network} on $X$. Then we say that a duet 
$\duet{x}{y}$ is {\bf induced by $N$} if there exists no triplet $t\in \cR(N)$ 
{such that $x,y\in L(t)$} and
the path in $U(N)$ joining $x$ and $y$ crosses precisely one degree two vertex of $U(N)$. For example, the duet $\duet{1}{4}$ is induced by the {$m$-network in the form of the} arboreal network in Fig.~\ref{fig:stackfree-necessary}(d).
The arboreal networks in Fig.~\ref{fig:not-encoded} demonstrate that there exist $m$-networks $N$ for which $\cD(N)$ is empty. 

We next extend the notion of an encoding of a rooted phylogenetic tree to 
$m$-networks {$N$ in a more restrictive way: in addition to making the same replacements as before we  also replace ``$\cR(N)$" by ``$\cR(N)\cup \cD(N)$".}
As the examples in {Fig.~\ref{fig:stackfree-necessary} 
demonstrate}, {even} $m$-rooted networks {that are arboreal} are, in general, not encoded by $\cR(N)\cup\cD(N)$. {For this example, it is due to $N^e$ containing a {\bf stacked hybrid}, that is, an arc whose head and tail is a hybrid.}
We call {an arboreal network $N$}  {\bf stack-free} if $N$ does not contain {such an arc.} 
Furthermore, 
we say that $N$ is {\bf banyan} if
the parents of every hybrid of $N$ are roots of $N$ and every leaf of $N$ is the child of a root or the child of a hybrid. For example, the arboreal network in Fig.~\ref{fig:mountain-range}(f) is banyan. 
{Denoting the class of stack-free arboreal networks on $X$ by $\cA(X)$, we have: }

\begin{lemma}\label{lem:encode-banyan}
{Let $N$ be a network in $\cA(X)$.}
Then the following are equivalent.
\begin{enumerate}
\item[(i)] $N$ is banyan.
\item[(ii)] $\cR(N)=\emptyset$ .
\item[(iii)]  $N$ is encoded by $\cD(N)$ {within $\cA(X)$.}
\end{enumerate}
\end{lemma}
\begin{proof}
(i) $ \Rightarrow $ (ii): Suppose that $N$ is banyan. 
Then, every leaf is the child of a root or of a hybrid.  For $N$ to induce a triplet $xy|z$ with $x\ne y\ne z$ in $X$, the shared parent $p$ of $x$ and $y$ cannot be a root.  So, $p$
must be a hybrid, but this contradicts the fact that the child of a hybrid is a leaf.
Thus, $\cR(N)=\emptyset$.

(ii) $ \Rightarrow $ (i): Assume that $\cR(N)=\emptyset$ and, for contradiction, that $N$ is not banyan. Then there exists a leaf $l$ of $N$ that is neither the child of a hybrid nor of a root of $N$, or there exist a hybrid $h$ of $N$ that has a parent that is not a root of $N$. In the former, it follows that $l$'s parent $p$  must be a tree vertex of $N$. Choose a vertex $v$ strictly above $p$ and leaves $l_v$ and $l_p$ such that $l_v$ is below $v$ but not below $p$ and $l_p$ is below $p$ but distinct from $l$. Then $l_pl|l_v\in \cR(N)$; a contradiction as $\cR(N)=\emptyset$. For the latter case, assume that $p$ is a parent of $h$ that is not a root of $N$. Then since $N$ is stack-free, 
$p$ must be a tree vertex. Choosing $v$ {to be the parent of $p$ and} $l_v$, and $l_p$ as before and denoting by $l$ a leaf below $h$ implies again $\cR(N)\not=\emptyset$; a contradiction. 

(i) \& (ii) $\Rightarrow$ (iii); Suppose that $N$ is banyan and that there exists a {network $N'\not=N$ in $\cA(X)$}
such that $\cR(N')\cup\cD(N')=\cD(N)\cup\mathcal R(N)$. Then, by (ii), $\cR(N)=\emptyset$. Since $\cR(M)\cap \cD(M)=\emptyset$ holds for any arboreal network $M$, it follows that $\cR(N')=\emptyset$. 
By the equivalence of {(i) and (ii)}, $N'$ must also be banyan {as it is stack-free}. Since 
$\cD(N')=\cR(N')\cup\cD(N')=\cR(N)\cup\cD(N)=\cD(N)$, it follows that every element in $X$ must be contained in the same number of duets in $\cD(N)$ as in $\cD(N')$. For all $x\in X$, we therefore have that $x$ is either adjacent to a root in $N$ if and  only if $x$ is adjacent to a root in $N'$ or that $x$ is adjacent to a hybrid in $N$ if and only if $x$ is adjacent to a hybrid in $N'$. Since the length of a path from a root to a leaf in a banyan network is one or two, it follows that there exists a bijection from $V(N)$ to $V(N')$ that induces an isomorphism between $N$ and $N'$. Thus, $N$ is encoded by $\cD(N)$.

(iii) $\Rightarrow$ (ii): Assume that $N$ is encoded by $\cD(N)$. Then the definition of a duet combined with the fact that {$X=L(\cD(N))$}
implies that $\cR(N)=\emptyset$. 
\end{proof}

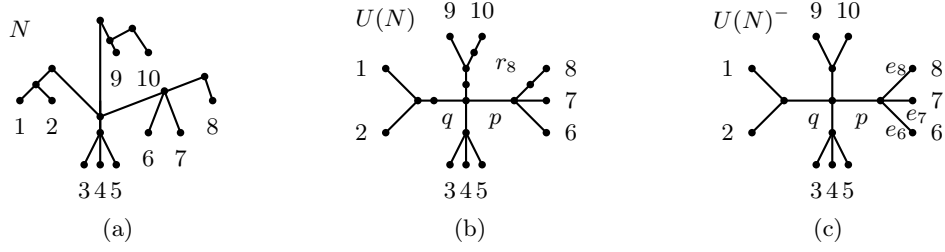
\begin{figure*}[t]
{\small 
\begin{center}
    \begin{tabular}{ccccc}
        \begin{tikzpicture}[thick,scale=\FigThreeScale]
    \node[align=left] at (1,5.25) {\small $N$};
    \node[outer sep=2pt, fill=black,circle,inner sep=1pt, label=below: {$1$} ] at (1,3){};
    \node[outer sep=2pt, fill=black,circle,inner sep=1pt, label=below: {$2$} ] at (2,3){};
    \node[outer sep=2pt, fill=black,circle,inner sep=1pt, label=below: {$3$} ] at (3,1){};
    \node[outer sep=2pt, fill=black,circle,inner sep=1pt, label=below: {$4$} ] at (3.5,1){};
    \node[outer sep=2pt, fill=black,circle,inner sep=1pt, label=below: {$5$} ] at (4,1){};
    \node[outer sep=2pt, fill=black,circle,inner sep=1pt, label=below: {$6$} ] at (5,2){};
    \node[outer sep=2pt, fill=black,circle,inner sep=1pt, label=below: {$7$} ] at (6,2){};
    \node[outer sep=2pt, fill=black,circle,inner sep=1pt, label=below: {$8$} ] at (7,3){};
    \node[outer sep=2pt, fill=black,circle,inner sep=1pt, label=below: {$9$} ] at (4,4.5){};
    \node[outer sep=2pt, fill=black,circle,inner sep=1pt, label=below: {$10$} ] at (5,4.5){};
    \node[fill=black,circle,inner sep=1pt] at (1.5,3.5){};
    \node[fill=black,circle,inner sep=1pt] at (2,4){};
    \node[fill=black,circle,inner sep=1pt] at (3.5,2.5){};
    \node[fill=black,circle,inner sep=1pt] at (3.5,2){};
    \node[fill=black,circle,inner sep=1pt] at (3.5,5.5){};
    \node[fill=black,circle,inner sep=1pt] at (6.75,3.75){};
    \node[fill=black,circle,inner sep=1pt] at (5.5,3.3){};
    \node[fill=black,circle,inner sep=1pt] at (4.5,5.25){};
    \node[fill=black,circle,inner sep=1pt] at (3.8,4.875){};
    \draw(1,3)--(1.5,3.5);
    \draw(2,3)--(1.5,3.5);
    \draw(1.5,3.5)--(2,4);
    \draw(2,4)--(3.5,2.5);
    \draw(3.5,2.5)--(6.75,3.75);
    \draw(3.5,2)--(3.5,3);
    \draw(3.5,2)--(3.5,5.5);
    \draw(4,4.5)--(3.5,5.5);
    \draw(4.5,5.25)--(5,4.5);
    \draw(3.8,4.875)--(4.5,5.25);
    \draw(6.75,3.75)--(7,3);
    \draw(5.5,3.3)--(5,2);
    \draw(5.5,3.3)--(6,2);
    \draw(3,1)--(3.5,2);
    \draw(3.5,1)--(3.5,2);
    \draw(3.5,2)--(4,1);    
\end{tikzpicture}
        &\mbox{\hspace{.25in}}
        &\begin{tikzpicture}[thick,scale=\FigThreeScale]
    \node[align=left] at (1,5.5) {\small $U(N)$};
    \node[outer sep=2pt, fill=black,circle,inner sep=1pt, label=left: {$1$} ] at (1,4){};
    \node[outer sep=2pt, fill=black,circle,inner sep=1pt, label=left: {$2$} ] at (1,2){};
    \node[outer sep=2pt, fill=black,circle,inner sep=1pt, label=below: {$3$} ] at (3,1){};
    \node[outer sep=2pt, fill=black,circle,inner sep=1pt, label=below: {$4$} ] at (3.5,1){};
    \node[outer sep=2pt, fill=black,circle,inner sep=1pt, label=below: {$5$} ] at (4,1){};
    \node[outer sep=2pt, fill=black,circle,inner sep=1pt, label=right: {$6$} ] at (6,2){};
    \node[outer sep=2pt, fill=black,circle,inner sep=1pt, label=right: {$7$} ] at (6,3){};
    \node[outer sep=2pt, fill=black,circle,inner sep=1pt, label=right: {$8$} ] at (6,4){};
    \node[outer sep=2pt, fill=black,circle,inner sep=1pt, label=above: {$9$} ] at (3,5){};
    \node[outer sep=2pt, fill=black,circle,inner sep=1pt, label=above: {$10$} ] at (4,5){};
    \node[fill=black,circle,inner sep=1pt,] at (2,3){};
    \node[fill=black,circle,inner sep=1pt,label=below left: {$p$} ] at (5,3){};
    \node[fill=black,circle,inner sep=1pt,label=below left: {$q$} ] at (3.5,3){};
    \node[fill=black,circle,inner sep=1pt] at (3.5,2){};
    \node[fill=black,circle,inner sep=1pt] at (3.5,4){};
    \node[fill=black,circle,inner sep=1pt] at (3.5,3.5){};
    \node[fill=black,circle,inner sep=1pt] at (2.5,3){};
    \node[fill=black,circle,inner sep=1pt] at (3.75,4.5){};
    \node[fill=black,circle,inner sep=1pt,label= above left: {$r_8$} ] at (5.5,3.5){};
    \draw(1,4)--(2,3);
    \draw(1,2)--(2,3);
    \draw(2,3)--(5,3);
    \draw(3.5,3)--(3.5,2);
    \draw(3,1)--(3.5,2);
    \draw(3.5,1)--(3.5,2);
    \draw(3.5,2)--(4,1);
    \draw(5,3)--(6,2);
    \draw(5,3)--(6,3);
    \draw(5,3)--(6,4);
    \draw(3,5)--(3.5,4);
    \draw(4,5)--(3.5,4);
    \draw(3.5,4)--(3.5,3);
\end{tikzpicture}
        &\mbox{\hspace{.25in}}
        &\begin{tikzpicture}[thick,scale=\FigThreeScale]
    \node[align=left] at (1,5.5) {\small $U(N)^-$};
    \node[outer sep=2pt, fill=black,circle,inner sep=1pt, label=left: {$1$} ] at (1,4){};
    \node[outer sep=2pt, fill=black,circle,inner sep=1pt, label=left: {$2$} ] at (1,2){};
    \node[outer sep=2pt, fill=black,circle,inner sep=1pt, label=below: {$3$} ] at (3,1){};
    \node[outer sep=2pt, fill=black,circle,inner sep=1pt, label=below: {$4$} ] at (3.5,1){};
    \node[outer sep=2pt, fill=black,circle,inner sep=1pt, label=below: {$5$} ] at (4,1){};
    \node[outer sep=2pt, fill=black,circle,inner sep=1pt, label=right: {$6$} ] at (6,2){};
    \node[outer sep=2pt, fill=black,circle,inner sep=1pt, label=right: {$7$} ] at (6,3){};
    \node[outer sep=2pt, fill=black,circle,inner sep=1pt, label=right: {$8$} ] at (6,4){};
    \node[outer sep=2pt, fill=black,circle,inner sep=1pt, label=above: {$9$} ] at (3,5){};
    \node[outer sep=2pt, fill=black,circle,inner sep=1pt, label=above: {$10$} ] at (4,5){};
    \node[fill=black,circle,inner sep=1pt,] at (2,3){};
    \node[fill=black,circle,inner sep=1pt,label=below left: {$p$} ] at (5,3){};
    \node[fill=black,circle,inner sep=1pt,label=below left: {$q$} ] at (3.5,3){};
    \node[fill=black,circle,inner sep=1pt] at (3.5,2){};
    \node[fill=black,circle,inner sep=1pt] at (3.5,4){};
    \draw(1,4)--(2,3);
    \draw(1,2)--(2,3);
    \draw(2,3)--(5,3);
    \draw(3.5,3)--(3.5,2);
    \draw(3,1)--(3.5,2);
    \draw(3.5,1)--(3.5,2);
    \draw(3.5,2)--(4,1);
    \draw(5,3)--node[below]{$e_6$}(6,2);
    \draw(5,3)--node[below right]{$e_7$}(6,3);
    \draw(5,3)--node[above]{$e_8$}(6,4);
    \draw(3,5)--(3.5,4);
    \draw(4,5)--(3.5,4);
    \draw(3.5,4)--(3.5,3);
\end{tikzpicture}\\
        (a) && (b) && (c)\\
    \end{tabular}
\end{center}
}
\caption{(a) An arboreal network $N$ on $X=\{1,\ldots,10\}$ considered in the case of $|S|\geq 2$ in the proof of Theorem~\ref{the:encoding}. (b) The tree $U(N)$ on $X$. (c) The phylogenetic tree $U(N)^-$ on $X$. For the vertex $p$ of $U(N)^-$ as indicated, $Y=\{6,7,8\}$ and 
$S=\{6,7\}$.}
\label{fig:encoding-general}

\end{figure*}

Fig.~\ref{fig:stackfree-necessary}(e,f) shows the necessity that the arboreal network be stack-free for Lemma~\ref{lem:encode-banyan}  to hold.  

\begin{theorem} \label{the:encoding}
 {A network $N$ in the class $\cA(X)$ is encoded by $\cD(N)\cup \cR(N)$ within that class.}
\end{theorem}

\begin{proof}
{We use induction on $n=|X|$.}
{The base case:} $n=3$. Then $N$ has either one or two roots. If $N$ has one root, then $N$  is a rooted triplet, since a root has out-degree 2. Clearly, $N$ is encoded by $\cR(N)$ in this case.
If $N$ has two roots, then $N$ is banyan. By Lemma~\ref{lem:encode-banyan}, $N$ is encoded by $\cD(N)$.

{Inductive Step: Assume true for $n$ and show true for $n+1$.}
{Let $|X|=n+1$ and assume that $N\in\cA(X)$.} 
In view of Lemma~\ref{lem:encode-banyan}, we may assume that $\cR(N)\not=\emptyset$ since otherwise $N$ is banyan and is encoded by $\cD(N)=\cD(N) \cup\cR(N)$.
Since $\cR(N)\not=\emptyset$, it follows that $U(N)^-$  cannot be a star tree.  
Choose an interior vertex $p$ of $U(N)^-$ such that $p$ is adjacent to every element in a generalized cherry $Y$ and precisely one interior vertex in $U(N)^-$, {denoted $q$}. 
Note that $q$ must exist 
{since $U(N)^-$ cannot be a star tree.}
Let $S$ denote the set of elements $y\in Y$ such that $\{y,p\}$ is also an edge in $U(N)$. 
We perform a case analysis on $|S|$ -- see Fig.~\ref{fig:encoding-general} for an illustration of the case $|S|\ge 2$.  
For all $y\in Y$, let $e_y$ denote the edge $\{p,y\}$ in $U(N)^-$ and, for all $y\in Y-S$,  let $r_y$ denote the subdivision vertex of $e_y$ in $U(N)$. Note that $r_y$ is a root of $N$ for all such $y$.

{\bf Case $|S|=0$:} Then, for all $y\in Y$, the  edge $e_y$ of $U(N)^-$ is subdivided by the vertex $r_y$ in $U(N)$.  Hence, $r_y$ is a root of $N$ for all such $y$. Thus, $p$ must be a hybrid in $N$. Since the out-degree of a hybrid of $N$ is 1 it follows that $q$ must be the child of $p$ in $N$. Furthermore, the fact that $N$ is in $\cA(X)$
combined with the fact that  $q$ is an interior vertex of $U(N)^-$, implies that $q$ is a tree vertex of $N$. Choose some element $x\in Y$. {Note that $x\not=q$ because $q\not\in Y$.}
Let $N'$ denote the mDAG
with leaf set ${X':=}X-\{x\}$ obtained from $N$ by deleting the vertices $x$ and $r_x$ and the arcs $(r_x, x)$ and $(r_x,p)$ from $N$ (and suppressing $p$ if this has rendered it a vertex with in-degree 1 and out-degree 1). Since $|Y|\geq 2$ it follows that  {$N'\in\cA(X')$}
since  
{$N\in\cA(X)$.}
By induction {hypothesis}, $N'$ is encoded by $\cR(N')\cup \cD(N')$.
Since $q$ is an interior vertex of $U(N)^-$, we have $\cD(N')=\cD(N)$ and 
$\cR(N')=\cR(N)- \{t\in\cR(N)\,:\, x\in L(t)\}$. Since $N$ can be recovered from {$N'$} by adding the vertices $x$ and $r_x$ and the arcs $(r_x,x)$ and $(r_x,p)$ to $N'$ in case $p$ was not suppressed in the construction of $N'$ from $N$ and this only adds triplets in $\cR(N)$ whose leaf set contains $x$, it follows that $N$ is also encoded by $\cR(N)\cup \cD(N)$. If $p$ was suppressed in the construction of $N'$ from $N$, we first {subdivide} the incoming arc of $q$ in $N'$ by a new vertex $p'$ and add the vertices $x$ and $r_x$ and the arcs $(r_x,x)$ and $(r_x,p')$ to $N'$. Similar arguments imply that $N$ is encoded by 
 $\cR(N)\cup \cD(N)$.

{\bf Case $|S|=1$:} Then $p$ must be a hybrid of $N$ and the sole element $s\in S$ is the child of $p$ in $N$ since the out-degree of a hybrid is 1. Choose an element {$x\in Y-S$} which must exist since $|Y|\geq 2$.
Constructing $N'$ from $N$ as {before}
implies that $N'$ is 
{a network in $\cA(X-\{x\})$} 
that is encoded by $\cR(N')\cup \cD(N')$.
Note that contrary to the previous case, $\cD(N')=\cD(N)-\{\duet{x}{s}\}$ and 
$\cR(N')=\cR(N)$. {Clearly,} $N$ can be recovered from $N'$ by adding the vertices $x$ and $r_x$ and the arcs $(r_x, x)$ and $(r_x,p)$ 
in case $p$ was not suppressed in the construction of $N'$ and by subdividing the incoming arc of $s$ in $N'$ by a vertex $p'$ and adding the vertices $x$ and $r_x$ and the arcs  $(r_x, x)$ and $(r_x,p')$ {otherwise}. {Thus,} $N$ is again encoded by $\cR(N)\cup \cD(N)$ since this only adds duets {from} $\cD(N)$  to  $\cD(N')$ that contain $x$.

\begin{figure*}[h]
\hrule
\noindent
{\sc Arboreal Reconstruction Algorithm {(ARA)}:}\\
{\sc Input:} A triplet system, $\cR$, and a duet system, $\cD$, on $X := L(\cR) \cup L(\cD)$ with $n := |X|$.\\
{\sc Output:} A {network $N\in \cA(X)$} with $\cR = \cR(N)$ and $\cD = \cD(N)$ or None, if no network exists.
\hrule
\noindent
\begin{tabular}{r p{4.52in} cr}
1.& Partition $X$
corresponding to disjoint network regions (see Sec.~\ref{sec:preprocess}). &&
    $O(|\cR|+|\cD|) = O(n^3)$\\
2.& Build a scaffold linking the roots to the elements in $X$
(see Sec.~\ref{sec:scaffold_tree} \& Fig.~\ref{fig:scaffold_alg}).&&
    $O(n^3)$\\
3.& For each block in the partition, iteratively refine the tree under each root in that block, merging common subtrees (see Sec.~\ref{sec:refine_partitions} and Fig.~\ref{fig:refine_alg}). &&
    $O(n^3)$\\
4. & Return the network, or None, if no network is possible (see Sec.~\ref{sec:correctness}).&&
    $O(1)$\\
\end{tabular}
\hrule
    \caption{Outline of the {\sc Arboreal Reconstruction Algorithm} {({\sc ARA})} with running time complexity in the right margin.}
    \label{fig:build_alg}
\end{figure*}

{\bf Case $|S|\geq 2$:} Then $p$ must be a tree vertex of $N$. Hence, one of the following two cases must hold: (i) $S=Y$ and so the arc $(q,p)$ or the arcs $a_p:=(r,p)$ and $a_q:=(r,q)$ are contained in $N$ where $r$ is a root of $N$ or (ii) $|S|=|Y|-1$ and $(p,q)$ is an arc in $N$. 

{\em Case~(i):} Choose some $x\in S$. {Then $x\not=q$ since $S\subseteq Y$ and $q\not\in Y$.}
If $|S|\geq 3$, then let $N'$ denote the mDAG
with leaf set ${X':=}X-\{x\}$ obtained from $N$ by deleting $x$ and the arc $(p,x)$. Note that the out-degree of $p$ in $N'$ is still at least 2. Then similar arguments as in the previous cases imply that $N'$ {is a network in $\cA(X')$} 
that is encoded by  $\cR(N')\cup \cD(N')$. Furthermore and independent of whether $(q,p)\in E(N)$ or {the arcs} $a_p$, $a_q$ {are} in $E(N)$, we have 
$\cD(N')=\cD(N)$ and 
$\cR(N')=\cR(N)- \{t\in\cR(N)\,:\, x\in L(t)\}$.
Since $N$ can be recovered from $N'$ by adding $x$ and the arc $(p,x)$, similar arguments as before imply that $N$ is encoded by  $\cR(N)\cup \cD(N)$. 

So assume that $|S|=2$. Let $N'$ denote the mDAG with leaf set ${X'}$ obtained from $N$ by deleting $x$ and the arc $(p,x)$ and suppressing $p$. Then, similar arguments as in the previous cases imply that $N'$ is a
{network in $\cA(X')$} that is encoded by  $\cR(N') \cup \cD(N')$. Let $y$ denote the other element in $S$.  Since $q$ is an interior vertex of $U(N)^-$ it must either be a tree vertex or a hybrid of $N$. 

If ${a_p,a_q}\in E(N)$, then if $q$ is a tree vertex then $\cD(N')=\cD(N)$ and
$\cR(N')=\cR(N)- \{t\in\cR(N)\,:\, x\in L(t)\}$. Otherwise, $q$ is a hybrid and so has a unique child $l$. 
In this case,
$\cD(N')=\cD(N)\cup\{\duet{l}{y}\}$ and
$\cR(N')=\cR(N)- \{l|xy\}$.
Similar arguments {as before} imply that $N$ is encoded by  $\cR(N)\cup \cD(N)$. 

Assume that $(q,p)\in E(N)$. If $q$ is a tree vertex, then we have $\cD(N')=\cD(N)$ and 
$\cR(N')=\cR(N)- \{t\in\cR(N)\,:\, x\in L(t)\}$. 
Similar arguments imply that $N$ is also encoded by  $\cR(N)\cup \cD(N)$. If $q$ is a hybrid of $N$, then $\cR(N')=\cR(N)- \{t\in\cR(N)\,:\, x\in L(t)\}$ holds again. Since $N$ is stack-free, a parent of $q$ in $N$ must either be a tree vertex or a root of $N$. Let $Q$ denote the set of parents of $q$ that are roots in $N$ and whose other child is a leaf of $N$. 
If $Q=\emptyset$ then $\cD(N')=\cD(N)$. Otherwise $\cD(N')=\cD(N)\cup \{\duet{d}{y}\,|\, d\in Q\}$.
Similarly, the result holds.

{\em Case~(ii):} Choose some $x\in S$. {As before, $x\not=y$.} Construct an mDAG $N'$ from $N$ as in the previous cases by removing $x$. Note that the out-degree of $p$ in $N'$ is still at least 2 because $q$ is a child of $p$. Using similar arguments {as before}, it follows that $N'$ is a 
{network in $\cA(X-\{x\})$} that is encoded by $\cR(N') \cup \cD(N')$. Furthermore, $\cD(N')=\cD(N)$ and 
$\cR(N')=\cR(N)- \{t\in\cR(N)\,:\, x\in L(t)\}$. Similar arguments imply that $N$ is  encoded by  $\cR(N)\cup \cD(N)$. 
\end{proof}

\begin{figure*}[t]
    \centering

\begin{tabular}{c c c c c }
    {\tiny
    \begin{tabular}{p{.4325in} | p{.2in}}
        triplets & duets\\
        \hline
        $12|3$, $12|4$, &\\
        $12|5$, 
        $34|1$, $34|2$, $34|6$, $34|7$, $34|8$, $34|9$,
        $35|1$, $35|2$, $35|6$, $35|7$, $35|8$, $35|9$,
        $45|1$, $45|2$, $45|6$, $45|7$, $45|8$, $45|9$,
        $36|8$, $37|9$, $46|8$, $47|9$, $56|8$, $57|9$,
        $67|8$
        &  $\duet{9}{10}$\\
    \end{tabular}
    }
    & \begin{tabular}{c}
    \begin{tikzpicture}[thick,scale=\FigFourScale]   
    \node[outer sep=2pt, fill=black,circle,inner sep=1pt, label=below: {\scriptsize $1$} ] (1) at (1,3){};
    \node[outer sep=2pt, fill=black,circle,inner sep=1pt, label=below: {\scriptsize $2$} ] (2) at (2,3){};
    \node[outer sep=2pt, fill=black,circle,inner sep=1pt, label=below: {\scriptsize $3$} ] (3) at (3,1){};
    \node[outer sep=2pt, fill=black,circle,inner sep=1pt, label=below: {\scriptsize $4$} ] (4) at (3.5,1){};
    \node[outer sep=2pt, fill=black,circle,inner sep=1pt, label=below: {\scriptsize $5$} ] (5) at (4,1){};
    \node[outer sep=2pt, fill=black,circle,inner sep=1pt, label=below: {\scriptsize $6$} ] (6) at (5,2){};
    \node[outer sep=2pt, fill=black,circle,inner sep=1pt, label=below: {\scriptsize $7$} ] (7) at (6,2){};
    \node[outer sep=2pt, fill=black,circle,inner sep=1pt, label=below: {\scriptsize $8$} ] (8) at (7,3){};
    \node[outer sep=2pt, fill=black,circle,inner sep=1pt, label=below: {\scriptsize $9$} ] (9) at (4.75,4.25){};
    \node[outer sep=2pt, fill=black,circle,inner sep=1pt, label=below: {\scriptsize $10$} ] (10) at (7,7.75){};
    \node[align=left] (PP) at (0.5,5) {\scriptsize $P$};
    \node[align=left] (P) at (6.5,8.5) {\scriptsize $P'$};
    \node[draw, opacity=0.1, fill=gray, rounded corners, inner sep=9pt, 
    fit= (1) (2) (3) (4) (5) (6) (7) (8) (9)] {};
    \node[draw, opacity=0.1, fill=gray, rounded corners, inner sep=9pt, 
    fit=(10)] {};
\end{tikzpicture}
    \end{tabular}
    &
    \begin{tabular}{c}
    \begin{tikzpicture}[thick,scale=\FigFourScale]   
    \node[outer sep=0pt, fill=black,circle,inner sep=1pt, label=below: {\scriptsize $1$} ] (1) at (1,3){};
    \node[outer sep=0pt, fill=black,circle,inner sep=1pt, label=below: {\scriptsize $2$} ] (2) at (2,3){};
    \node[outer sep=0pt, fill=black,circle,inner sep=1pt, label=below: {\scriptsize $3$} ] (3) at (3,1){};
    \node[outer sep=0pt, fill=black,circle,inner sep=1pt, label=below: {\scriptsize $4$} ] (4) at (3.5,1){};
    \node[outer sep=0pt, fill=black,circle,inner sep=1pt, label=below: {\scriptsize $5$} ] (5) at (4,1){};
    \node[outer sep=0pt, fill=black,circle,inner sep=1pt, label=below: {\scriptsize $6$} ] (6) at (5,2){};
    \node[outer sep=0pt, fill=black,circle,inner sep=1pt, label=below: {\scriptsize $7$} ] (7) at (6,2){};
    \node[outer sep=0pt, fill=black,circle,inner sep=1pt, label=below: {\scriptsize $8$} ] (8) at (7,3){};
    \node[outer sep=0pt, fill=black,circle,inner sep=1pt, label=right: {\scriptsize $9$} ] (9) at (4.75,4.25){};
    \node[outer sep=0pt, fill=black,circle,inner sep=1pt, label=right: {\scriptsize $10$} ] (10) at (7,7.75){};
    \node[align=left] (PP) at (0.5,5) {\scriptsize $P$};
    \node[align=left] (P) at (6.5,8.5) {\scriptsize $P'$};
    \node[draw, opacity=0.1, fill=gray, rounded corners, inner sep=9pt, 
    fit= (1) (2) (3) (4) (5) (6) (7) (8) (9)] {};
    \node[draw, opacity=0.2, fill=gray, rounded corners, inner sep=2pt, 
    fit= (1) (2)] {};
    \node[draw, opacity=0.2, fill=gray, rounded corners, inner sep=2pt, 
    fit= (3) (4) (5) (6) (7)] {};
    \node[draw, opacity=0.2, fill=gray, rounded corners, inner sep=2pt, 
    fit= (8)] {};
    \node[draw, opacity=0.2, fill=gray, rounded corners, inner sep=2pt, 
    fit= (9)] {};
    \node[draw, opacity=0.1, fill=gray, rounded corners, inner sep=9pt, 
    fit=(10)] {};
    \draw(10)--(9) node [midway,above] {\tiny $r_0$}; 
    \draw(2.325,3.25)--(3.5,2.5) node [midway,above] {\tiny $r_1$}; 
    \draw(4.125,2.5)--(4.5,3.9) node [midway,right] {\tiny $r_2$}; 
    \draw(5.5,2.5)--(6.9,2.9) node [midway,above] {\tiny $r_3$}; 
       
\end{tikzpicture}
    \end{tabular}
    &
    \begin{tabular}{c}
    \begin{tikzpicture}[thick,scale=\FigFourScale]   
    \node[outer sep=2pt, fill=black,circle,inner sep=1pt, label=below: {\scriptsize $1$} ] (1) at (1,3){};
    \node[outer sep=2pt, fill=black,circle,inner sep=1pt, label=below: {\scriptsize $2$} ] (2) at (2,3){};
    \node[outer sep=2pt, fill=black,circle,inner sep=1pt, label=below: {\scriptsize $3$} ] (3) at (3,1){};
    \node[outer sep=2pt, fill=black,circle,inner sep=1pt, label=below: {\scriptsize $4$} ] (4) at (3.5,1){};
    \node[outer sep=2pt, fill=black,circle,inner sep=1pt, label=below: {\scriptsize $5$} ] (5) at (4,1){};
    \node[outer sep=2pt, fill=black,circle,inner sep=1pt, label=below: {\scriptsize $6$} ] (6) at (5,2){};
    \node[outer sep=2pt, fill=black,circle,inner sep=1pt, label=below: {\scriptsize $7$} ] (7) at (6,2){};
    \node[outer sep=2pt, fill=black,circle,inner sep=1pt, label=below: {\scriptsize $8$} ] (8) at (7,3){};
    \node[outer sep=2pt, fill=black,circle,inner sep=1pt, label=below: {\scriptsize $9$} ] (9) at (4.75,5.25){};
    \node[outer sep=2pt, fill=black,circle,inner sep=1pt, label=below: {\scriptsize $10$} ] (10) at (7,7.75){};
    \node[outer sep=2pt, fill=black,circle,inner sep=1pt, label=above: {\scriptsize $r_0$} ] at (5,8.5){};
    \node[draw, opacity=0.2, fill=gray, rounded corners, inner sep=2pt, 
    fit= (3) (4) (5) (6) (7)] {};
    \node[outer sep=2pt, fill=black,circle,inner sep=1pt] at (4.75,6){};
    \draw(5,8.5)--(7,7.75);
    \draw(5,8.5)--(4.75,5.75);
    \draw(4.75,6)--(4.75,5.25);
    \node[fill=black,circle,inner sep=1pt, label=right: {\tiny $r_1$}] at (1.75,5.25){};
    \draw(1.75,5.25)--(1.5,3.5);
    \draw(1.75,5.25)--(3.5,2.5);
    \node[fill=black,circle,inner sep=1pt] at (1.5,3.5){};
    \draw(1,3)--(1.5,3.5);
    \draw(2,3)--(1.5,3.5);
    \node[fill=black,circle,inner sep=1pt, label=above: {\tiny $r_2$}] at (3.75,6.75){};
    \draw(3.75,6.75)--(3.75,2.5);
    \draw(3.75,6.75)--(4.75,6);
    \node[fill=black,circle,inner sep=1pt, label=right: {\tiny $r_3$}] at (6.75,5.25){};
    \draw(6.75,5.25)--(5.5,2.5);
    \draw(6.75,5.25)--(7,3);
\end{tikzpicture}
    \end{tabular}
    &
    \begin{tabular}{c}
    \begin{tikzpicture}[thick,scale=\FigFourScale]   
    \node[outer sep=2pt, fill=black,circle,inner sep=1pt, label=below: {\scriptsize $1$} ] (1) at (1,3){};
    \node[outer sep=2pt, fill=black,circle,inner sep=1pt, label=below: {\scriptsize $2$} ] (2) at (2,3){};
    \node[outer sep=2pt, fill=black,circle,inner sep=1pt, label=below: {\scriptsize $3$} ] (3) at (3,1){};
    \node[outer sep=2pt, fill=black,circle,inner sep=1pt, label=below: {\scriptsize $4$} ] (4) at (3.5,1){};
    \node[outer sep=2pt, fill=black,circle,inner sep=1pt, label=below: {\scriptsize $5$} ] (5) at (4,1){};
    \node[outer sep=2pt, fill=black,circle,inner sep=1pt, label=below: {\scriptsize $6$} ] (6) at (5,2){};
    \node[outer sep=2pt, fill=black,circle,inner sep=1pt, label=below: {\scriptsize $7$} ] (7) at (6,2){};
    \node[outer sep=2pt, fill=black,circle,inner sep=1pt, label=below: {\scriptsize $8$} ] (8) at (7,3){};
    \node[outer sep=2pt, fill=black,circle,inner sep=1pt, label= below: {\scriptsize $9$} ] (9) at (4.75,5.25){};
    \node[outer sep=2pt, fill=black,circle,inner sep=1pt, label=below: {\scriptsize $10$} ] (10) at (7,7.75){};
    \node[outer sep=2pt, fill=black,circle,inner sep=1pt, label=above: {\scriptsize $r_0$} ] at (5,8.5){};
    \node[draw, opacity=0.2, fill=gray, rounded corners, inner sep=2pt, 
    fit= (3) (4) (5) (6) (7)] {};
    \node[outer sep=2pt, fill=black,circle,inner sep=1pt] at (4.75,6){};
    \draw(5,8.5)--(7,7.75);
    \draw(5,8.5)--(4.75,5.75);
    \draw(4.75,6)--(4.75,5.25);
    \node[fill=black,circle,inner sep=1pt, label=right: {\tiny $r_1$}] at (1.75,5.25){};
    \draw(1.75,5.25)--(1.5,3.5);
    \draw(1.75,5.25)--(3.5,2.5);
    \node[fill=black,circle,inner sep=1pt] at (1.5,3.5){};
    \draw(1,3)--(1.5,3.5);
    \draw(2,3)--(1.5,3.5);
    \node[fill=black,circle,inner sep=1pt, label=above: {\tiny $r_2$}] at (3.75,6.75){};
    \draw(3.75,6.75)--(3.5,2.5);
    \draw(3.75,6.75)--(4.75,6);
    \node[fill=black,circle,inner sep=1pt, label=right: {\tiny $r_3$}] at (6.75,5.25){};
    \draw(6.75,5.25)--(5.5,2.5);
    \draw(6.75,5.25)--(7,3);
    \node[fill=black,circle,inner sep=1pt] at (3.5,2.5){};
    \node[fill=black,circle,inner sep=1pt] at (3.5,2){};
    \draw(3.5,2.5)--(3.5,2);
    \draw(3,1)--(3.5,2);
    \draw(3.5,1)--(3.5,2);
    \draw(4,1)--(3.5,2);
\end{tikzpicture}
    \end{tabular}\\
    
    (a) & (b) & (c) & (d) & (e)\\
\end{tabular}
    
    \caption{{\sc An {\sc ARA} Example}: 
        (a) The {systems $\cR(N)$ and $\cD(N)$ for the network $N$ from Fig.~\ref{fig:encoding-general}(a)}.
        {After Line 1:}  (b) The $(\cD(N),\cR(N))$-induced partition of the leaves with blocks $P$ and $P'$
        with $P = \{1,2,\ldots,9\}$ and $P' =  \{10\}$.
        {After Line 2:} (c) we have edges for each root induced by the duet ($r_0$) and the triplets ($r_1,r_2,r_3$). $P'$ has a single component while there are three components of $P$. {After Line 3:}
        (d) The components with one or two leaves refined{, and  
        (e) refining} the component with leaf set $\{3,4,5,6,7\}$ is done by reconstructing the structure below each root.  First, the tree on the {leaf set} below $r_1$, $L_{r_1} = \{3,4,5\}$, is constructed. Since $L_{r_1} = L_{r_2}$, {their subtrees} are identical. Merging results in the subtree below a hybrid.  Next, the tree under $r_3$ on $L_{r_3} = \{3,4,5,6,7\}$ is reconstructed and merged, {also} resulting in a hybrid. {Since all blocks have now been processed, the generated network (which is the network in Fig~\ref{fig:encoding-general}(a)) is returned in Line 4}.
    }
    \label{fig:build_alg_ex}
\end{figure*}
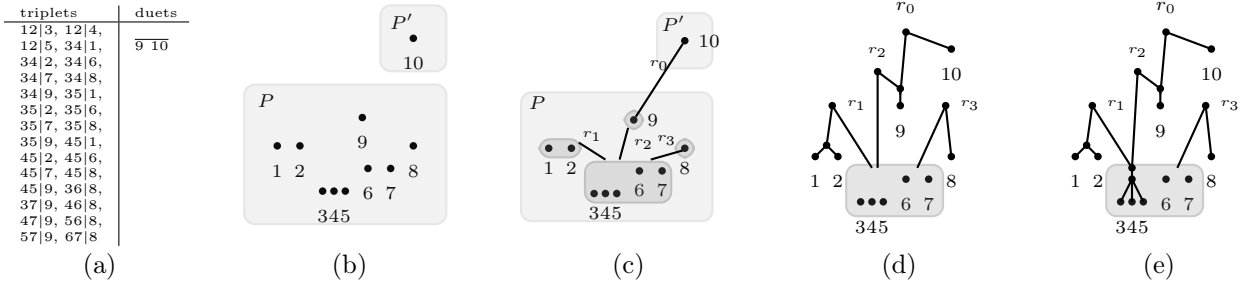

\section{{Reconstructing Arboreal Networks}}

{We present an algorithm, {{\sc ARA},} 
that efficiently reconstructs an arboreal network from them, if one exists, given the basic building blocks of duets and triplets.
It relies on our combinatorial classification (Theorem~\ref{the:encoding}) and the work of Bryant \cite{bryant}.}
Bryant \cite{bryant} used the {\sc Build} algorithm of Aho {\em et al.}~\cite{aho81} to determine, in polynomial time, if a triplet system is compatible with a rooted phylogenetic tree.  Bryant further showed that if the triplet system is not compatible, finding a maximal set of compatible triplets is computationally hard. The {\sc Build} algorithm first constructs a graph  whose components correspond to the subtrees of the root of {an} envisaged tree. 
If there is only one component, then the root would have only one child which is not possible, and the algorithm halts. 
Otherwise, the process is recursively applied to the subcomponents until the tree is resolved. If the triplets are not compatible, no tree is built and the algorithm halts.   A tree can have {\bf polytomies} i.\,e.\,vertices with out-degree larger than 3, {if  information about the 
relationships in the 
triplets is limited.}

We can show similar results for arboreal networks and duets and triplets, with a subtle difference.  In the case of arboreal networks, the absence of triplets and duets is itself a statement about the possible network.  For the triplet system $\mathcal C =\{12|3, 12|4\}$ in Fig.~\ref{fig:absense-of-evidence}(a), any binary rooted phylogenetic tree that induces 
$\mathcal C$ also induces one of $t_1 = 24|3$ or $t_2 = 23|4$. Thus, there exists no such
tree 
$T$ such that $\cR(T)=\mathcal C$. With multiple roots allowed, 
there is an arboreal network $N^b$ (Fig.~\ref{fig:absense-of-evidence}(b)) such that $\cR(N^b)= \mathcal C$.
Fig.~\ref{fig:build_alg} gives an overview of {\sc ARA}
with the details of each step provided in the subsequent sections and figures.
We illustrate {its inner workings}
in Fig.~\ref{fig:build_alg_ex}.

\subsection{{A} Partition Induced by Duets and Triplets}
\label{sec:preprocess}

Given a duet system $\cD$ and a triplet system $\cR$, we first {find a partition of
$X:=L(\cD) \cup L(\cR)$ whose}
blocks ({see the shaded areas in Figure~\ref{fig:build_alg_ex}(b)) {can be processed independently.}}
{To this end,} {we call a partition {$\mathcal P$}
of $X$ a {\em $(\cD,\cR)$-induced partition} if, for all duets {$\duet{x}{y}\in \cD$, there is no block in $\mathcal P$ that contains $\{x,y\}$, and for all triplets $t\in\cR$, there exist a block in $\mathcal P$ that contains $L(t)$}.   {Note that $\mathcal P$} is well-defined} since if $\duet{x}{y} \in \cD$ then, by definition of a duet, there is no $z\in X$ such that $xy|z\in\cR$.  
{We use this type of partition}
as a {framework} to the later subroutines that {make up {\sc ARA} and compute it in Proposition~\ref{prop:duet-induced}.}
{The blocks in the partition we then refine further} 
into components connected by edges that correspond to the roots in {the envisaged network}.

{The next result is a consequence of Theorem~\ref{the:encoding}.}

\begin{proposition}
    \label{prop:duets}
    Let $N$ be an arboreal network. Then, each duet in $\mathcal D(N)$ gives rise to a root of $N$.   
\end{proposition}

{Continuing with the notation in Fig.~\ref{fig:build_alg_ex}, a $(\cD,\cR)$-partition  can be found by an algorithm whose running time}
is dominated by $|\cR|$
which is at most $O(n^3)$:  

\begin{proposition}\label{prop:duet-induced}
    Let $\cR$ be a triplet system and let  $\cD$ be a duet system such that $X=L(\cR)\cup L(\cD)$. Then a $(\cD,\cR)$-induced partition 
    can be computed in $O(|\cR \cup \cD|) = O(|X|^3)$ time.
\end{proposition}

\begin{proof}
We first initialize several look-up tables for more efficient access.  We set up a table of triplets with values, $unused$, to keep track of which triplets we have used  and  a table, $belong$, of triplets in terms of leaves:  for each $xy|z\in\cR$, append $xy|z$ to $belong[x]$, to $belong[y]$, and to $belong[z]$.  The initialization takes $O(|\cR|) = O(|X|^3)$ and allows us to find the triplets in constant time in the steps below. We also initialize {to 0 }a table $which$ of the leaves. 

For each duet, $\duet{a}{b} \in \cD$, we set up queues $q_a = (a)$ and $q_b = (b)$.  We remove the first item, $x$, from $q_a$. We set $which[x] = a$ and add the leaves of all unused triplets involving $x$ to $q_a$, marking {them}
as used.  We repeat these steps until $q_a$ is empty.  We similarly work through $q_b$ with an extra check {to see} if a leaf has already been {assigned to}  $a$.  If so, we return None. 

{Note that after working through all duets, every leaf must  have been assigned, since every leaf must be a part of some triplet or duet.}
We process each duet and triplet once and the resulting {$(\cD,\cR)$}-partition of $X$ has $|\cD|+1$ blocks and can be determined in linear time in the size of $\cR$.
\end{proof}

{Since the underlying graph of an arboreal network is a tree, each part of the underlying graph is acyclic and connected.  As such, 
a necessary condition for a $(\cD,\cR)$-induced partition, ${\cP}$, to be induced by an arboreal network is that $\cD$ induces a connected, acyclic graph on the blocks of $\cP$.}

\subsection{Initializing the Scaffold}
\label{sec:scaffold_tree}

We build up the envisaged network using a tree 
as a scaffold.  This scaffold starts with the edges, representing the roots induced by 
the duets in $\cD$ and triplets in $\cR$.
To do this efficiently, we compute a look-up table of the roots that will be above each element in $X$, an augmented list of edges corresponding to the roots induced by the duets and by the triplets, and the induced components.

\begin{figure*}[h]
\hrule
{\small
\noindent
{\sc Scaffold Subroutine:}\\
{\sc Input:} Lists, $\cD$, of duets and $\cR$ of triplets, and a $(\cD,\cR)$-induced partition, $\mathcal P$, of $X:=L(\cD)\cup L(\cR)$.\\
{\sc Output:} A list, $comps$, of triplet systems, a list, $re$, of {root} edges  and an ancestor (look-up) table, $anc$, for $X$.
\hrule
\smallskip
\noindent
\begin{tabular}{r p{4.25in} crrr}
1.& Initialize $comps = \emptyset$, $E = \emptyset$, $re = \emptyset$, $\forall l \in X$, $anc[l]=\emptyset$.  &&&&{ $O(n)$}\\

2.& For each  $ab|c \in \cR$: Append $\{a,b\}$ to $E$.
&&&&$O(|\cR|)$\\
3.& For each  $\duet{a}{b}\in \cD$: Append $\{a,b\}$ to $E$, $re$, $anc(a)$, \& $anc(b)$. 
&&&&$O(|\cD|)$\\

4.& For each block $P \in \mathcal P$:
&&&&
$O(n^2)$\\
5.& \quad Construct graph, $G_P = (P,\{\{u,v\}\in E : u,v\in P\})$.&&& {\small $O(|P|^2)$}\\
6.&\quad If $G_P$ is connected and $|P|>1$, the construction halts, as in~\cite{aho81}.&&&{\small $O(|P|)$}\\
7.&\quad Else: Let $C_P = \{C_{P,1},\ldots,C_{P,k}\}$ be the component leaf sets of $G_P$.
&&&{\small $O(|P|)$}\\
8.&\quad\quad Append the set $C_P$ to $comps$.  &&&{\small $O(|P|)$}\\
9.&\quad\quad For each distinct $C,C'\in C_P$:  &&&
{\small $O(|P|^2)$}\\
10.&\quad\quad\quad If $ab|c \in\cR$, $a\in C$, and $c\in C'$, &&{\small $O(1)$}\\ 
11.&\quad\quad\quad\quad Add edge $e = \{C,C'\}$ to $re$.&&{\small $O(1)$}\\
12.&\quad\quad\quad\quad Append $e$ to $anc[a]$ and $anc[b]$.&&{\small $O(1)$}\\
13.& Return $re$, $comps$, and $anc$.&&&& $O(1)$\\
\end{tabular}
\hrule
}
    \caption{Outline of the {\sc Scaffold} subroutine with running time complexity $O(|\cD| + |\cR|) = O(n^3)$.  Bounds on the running time of each line is in the margin  with loops containing the complexity of all nested operations. }
    \label{fig:scaffold_alg}
\end{figure*}

We show for {a network, $N\in\cA(X)$}, that the {\sc Scaffold} subroutine {{returns the correct structure when given $\cD(N)$, $\cR(N)$ and 
the $(\cD(N),\cR(N))$-induced partition of $X$}.  More precisely,} 
we define the {\bf component graph}, {$G(N)$,} of an arboreal network $N$ as follows:  Let $C(N)$ be the partition of $X$
when all degree two vertices, and their adjacent edges are removed from $N$.  Let $EC(N)$ be the set of removed edges{, that is:} 
$\{C,C'\}{\subseteq C(N)}$ is an edge {precisely} if there is a degree 2 vertex that is adjacent to $C$ and $C'$.
Let $G(N) = (C(N),EC(N))$.
Note that we will use the ancestor table $anc$ built in Line 12 in the {\sc Refine} subroutine and show the correctness there.

\begin{proposition}
    Let $N$ be an arboreal network on $X$.  
    Let {$re$, $comps$}
    and $anc$ result from running the {\sc Scaffold} subroutine on the $(\cD(N),\cR(N))$-induced partition. 
    Let $G'=(comps, re)$. Then, $G(N) \simeq G'$.
    \label{prop:comp_graph}
\end{proposition}

\begin{proof} 
{We use} induction on the number $c$ of components of $G(N)$.
Let $\mathcal P(N)$ be the $(\cD(N),\cR(N))$-induced partition of $X$. 
The base case: $c = 2$. Then $N$ has a single root, and we have two cases.  The first case is that $N$ induces only a single duet, $\duet{a}{b}$ and no triplets.  The edge $r=\{a,b\}$ is added to $E$, $re$, $anc(a)$ and $anc(b)$ in Line 3 of Fig.~\ref{fig:scaffold_alg}, while $comps$ is updated to $\{\{a\},\{b\}\}$ in Line 8. The algorithm returns $comps = \{\{a\},\{b\}\}$, $re = \{r\}$, and the ancestor tables $anc[a]=r$ and $anc[b]=r$.  The resulting graph $G'=(comps, re) = (\{\{a\},\{b\}\},\{r\})$ is exactly $G(N)$. 
The second case is $|\cD(N)|=0$.  
Then $\mathcal P(N)$ has a single block, and $G(N)$ has two vertices $C$ and $C'$ and a single edge joining them. 
Then, $N$ is a rooted phylogenetic tree, and Line 2 of Fig.~\ref{fig:scaffold_alg} yields two maximal cliques on $X$. Thus, $G_P$ has two component leaf sets $C_{P,1}$ and $C_{P,2}$, and $\{C_{P,1}, C_{P,2}\}$ is appended to $comps$ in Line 8.  Since $|X|\ge 3$, there is a triplet that crosses the root of $N$. Thus, in Line 11, the edge $\{C_{P,1},C_{P,2}\}$ is added to $re$. Since both $G(N)$ and the graph $G=(comps,re)$ have two vertices and one edge, they are isomorphic.

Inductive Step:  Assume true for $c$ and show true for $c+1$. 
Again, we have two cases:  The first case is $|\cD(N)| > 0$. Then, $|\cP(N)| > 1$.  Let $\duet{a}{b}\in\cD(N)$ and {let} $N'$ and $N''$ be the restrictions of $N$ to each side of the duet. 
The number of components in $N$ and in $N'$ is at most $c$. By inductive hypothesis, $G(N')$ and $G(N'')$ are isomorphic to the graphs returned by {the} {\sc Scaffold} {subroutine}. 
Since the blocks are processed independently of each other in 
Fig.~\ref{fig:scaffold_alg}, the resulting graph is $G(N')$ and $G(N'')$ with the additional edge $\{a,b\}$ added (Line 3) {which is isomorphic with} $G(N)$.

The second case is $|\cD(N)| = 0$.  Then, $\cP(N)$ has a single block, $P$.  Since $N$ is an arboreal network, there exists a subtree, $T$, of $U(N)$ which is adjacent to a vertex of degree $2$.  Since $T$ and the 
arboreal network $N|_{X-L(T)}$ obtained by restricting $N$ to ${X-L(T)}$ 
both have at most $c$ components and the inductive hypothesis applies, the respective component graphs and the graphs resulting from Fig.~\ref{fig:scaffold_alg} are isomorphic.  By choice of $T$, there is a single edge connecting $T$ and  $N|_{X-L(T)}$ in $G(N)$  as well as constructed in Line 10, yielding the desired isomorphism.
\end{proof}

\subsection{Refining the Partition}
\label{sec:refine_partitions}

In the previous step, we decomposed each block of a $(\cD,\cR)$-induced partition of $X$ into components.  
By calling the {\sc Refine} subroutine on each component, we can reconstruct an arboreal network (if one exists). 
For each component, {$C$}, we iterate through the roots above {$C$,}
constructing the tree of descendants.  Given that 
{for an arboreal network $N$ the graph $U(N)^-$}
is a phylogenetic tree, if a set of leaves is below two different roots, then their overlap forms a subtree with identical structure under each root: 

\begin{figure*}[h]
\hrule
\noindent
{\small
{\sc Refine Subroutine:}\\
{\sc Input:} {A triplet system $\mathcal C$ with $m = |L(\mathcal C)|\leq n$}
and an {ancestor (look-up)} table $anc$.\\
{\sc Output:} The arc set, {\em arcs}, of an arboreal network $N$, on $\mathcal L:=L(\mathcal C) \cup values(anc)$ such that $\mathcal C=\mathcal R(N)$, or None if no network exists.
\smallskip
\hrule
\noindent
\begin{tabular}{r p{4.8in} crr}
0.& Initialize table $arcs$ to $\emptyset$.
&& &$O(1)$\\
1.& For each $r$ in the values of $anc$ table:&&&$O(|\mathcal C|)$
\\
2.&\quad Let $L_r = \{l\in {\mathcal L} 
: r \in anc[l]\}$&&{\small $O(m)$
}\\
3.&\quad Run {\sc Build}
on $\mathcal C$ restricted to $L_r$ resulting in {a rooted phylogenetic tree} $T(r)$. &&{\small $O(|\mathcal C|)$}\\
4.&\quad If the root of $T(r)$ does not have out-degree 2, return None.&&$O(1)$\\
5.&\quad Else: Add the arcs of $T(r)$ to $arcs$, merging common arcs.&&{\small $O(|L_r|)$}\\
6.&\quad\quad Add the arc $(r,r')$ to $arcs$ where $r'$ is the root of $T(r)$.&&{\small $O(1)$}\\
7.&Return $arcs.$&&&$O(1)$\\
\end{tabular}
\hrule
}
    \caption{Outline of the {\sc Refine} subroutine with running time complexity $O(|\mathcal C|) = O(m^3)$.
    Bounds on the running time of each line is in the margin  with loops containing the complexity of all nested operations. }
    \label{fig:refine_alg}
\end{figure*}

\begin{lemma}  
{Let $N$ be an arboreal network and let}
    $r$ and $r'$ be roots of $N$ {such that} $L(r) \cap L(r') \ne \emptyset$. Then there exists a vertex, $v$, on the {(undirected)} path from $r$ to $r'$ that is a hybrid below both $r$ and $r'$.  Further, when viewed as phylogenetic trees, the subtree of $N_r$ rooted at {the child $c$ of} $v$ is isomorphic with the subtree of $N_{r'}$ rooted at $c$.
    \label{lem:subtrees-same}
\end{lemma}

\begin{proof}
    Choose some $x\in L:=L(r) \cap L(r') $. Let $P_{rx}$ and $P_{r'x}$ denote the paths in $N$ from $r$ to $x$ and $r'$ to $x$, respectively. Since $x$ is a vertex in both paths, there must exist a vertex $v$ of $N$ that is the first vertex on $P_{rx}$ that also lies on $P_{r'x}$. By construction, $v$ has in-degree at least $2$ and so must be a hybrid of $N$.
    Let $K$ be the subtree of $N_r$ rooted at the child $c$ of $v$ and let $K'$ be the subtree of $N_{r'}$ rooted at $c$, with degree two vertices {suppressed} in each case to yield phylogenetic trees.
    Then $L(K)=L=L(K')$. 
    Since $\cR(K)$ {equals} the restriction of $\cR(N)$ to $L$ and similarly for $\cR(K')$,  $K$ and $K'$ must be isomorphic {phylogenetic trees}
    \cite[Theorem 6.4.1]{SS03}.
\end{proof}

\subsection{Correctness \& Running Time}
\label{sec:correctness}

Using the analysis above, we show that {{\sc ARA}}
is correct and runs in $\Theta(|\cR|) = \Theta(n^3)$.

\begin{lemma}  Let {$N\in \cA(X)$}. 
{Let} $N'$ be the network built by {\sc ARA}
from $\cR(N)$ and $\cD(N)$. {Then} $N \simeq N'$.
\end{lemma}

\begin{proof}   
    We rely on Theorem~\ref{the:encoding} that if $\cD(N)=\cD(N')$ and $\cR(N)=\cR(N')$, then {$N$ and $N'$} are isomorphic.
    By Proposition~\ref{prop:duets}, every duet corresponds to a root in $N$.  By construction,  a root 
    is added to $N'$ for each duet, yielding, $\cD(N)=\cD(N')$.
    
    To show that $\cR(N)=\cR(N')$, let $ab|c \in \cR(N)$. {We have two cases. The first case  that} $ab|c$ is {\em underneath} a root -- that is the undirected paths between the three leaves do not cross a root. 
    Let $r$ be that root. Then $a,b,c\in L(r)$. Let $C$ be the leaf set of a connected component of the graph $G_P$ (see Line 7 of the {\sc Scaffold} subroutine) 
    such that $a,b,c\in C$.
    To construct $N'$, the {\sc Refine} subroutine is run on all components.  
    In Line 3, the {\sc Build} algorithm \cite{aho81} is run on $\cR(N)$ 
    restricted to {the set} $L_r$ {defined in that line}, which generates a subtree {$T(r)$ also defined in that line} which induces all triplets {$t\in \cR(N)$ with $L(t)\subseteq C$.} 
    Thus, $ab|c$ is induced by the generated subtree. By Lemma~\ref{lem:subtrees-same}, $ab|c \in \cR(N')$.
    To show the reverse direction for triplets underneath a root, we note that if $ab|c\in\cR(N')$ and below all roots, then it was added to $N'$ in Line 6 of the {\sc Refine} subroutine.  For that to occur, there exists a root $r$ in $N$ such that $a, b, c\in L(r)\cap C$ for {$C$ the leaf set of} a component (see Line 5 of the {\sc Scaffold} subroutine). The {\sc Build} algorithm  on the triplets in 
    $\cR(N)$ underneath $r$ yields $ab|c$.  
    By the correctness of the {\sc Build} algorithm \cite{bryant}, 
    $ab|c \in \cR(N).$ Thus, $\cR(N)=\cR(N')$ in this case.

    {The second case is} that $ab|c\in\cR(N)$ crosses a root, $r$, of $N'$. Then
    there exist components with leaf sets $C$ and $C'$, with $C\ne C'$ and $a,b\in C$ and $c\in C'$ (Line 10 of the {\sc Scaffold} subroutine), and $C$ and $C'$ are both below $r$.  After applying the {\sc Refine} subroutine, each component below a root is a subtree of $U(N')^-$, $a,b\in C$ and $c\in C'$. Thus, $ab|c \in \cR(N')$.
    Lastly, consider the case that $ab|c\in\cR(N')$ and crosses a root, $r$, in $N'$. 
    As before, there exist components with leaf sets $C$ and $C'$, with $C\ne C'$ and $a,b\in C$ and $c\in C'$ (see Line 10 of the {\sc Scaffold} subroutine).  
    Similar arguments as before imply $ab|c\in \cR(N)$. Thus, $\cR(N)=\cR(N')$ holds in this case. 
\end{proof}

\begin{theorem}
    \label{thm:build}
    {Let $\cD$  be a duet system and $\cR$ a triplet system.} 
    {For} $X={ L}(\cD)\cup { L(}\cR)$,  
    we can {decide} in polynomial time in $|X|$, if there exists an arboreal network $N$ on $X$ such that $\cR(N)=\cR$ and $\cD(N)=\cD$.
\end{theorem}

\begin{proof}
The running time complexity is given line-by-line in Fig.~\ref{fig:build_alg}, \ref{fig:scaffold_alg}, and \ref{fig:refine_alg} and is polynomial in $|X|$.
\end{proof}

\section{A {Metric for Stack-free} Arboreal Networks}
\label{sec:metric}

Since triples are not enough to {encode} a 
network {in $\cA(X)$} but triples and duets are (Theorem~\ref{the:encoding}), we {obtain} a  metric  {on $\cA(X)$} that {canonically} extends the popular triplet distance \cite{critchlow1996triples} for rooted phylogenetic trees. 
{More precisely,  we} define the {{\bf duet-triplet distance}}
$d(N,N')$ {between networks $N, N'\in \cA(X)$}
as  
$|(\cD(N)\cup \cR(N))\Delta (\cD(N')\cup \cR(N'))|$ {where $\Delta$ is the symmetric distance (see Fig.~\ref{fig:absense-of-evidence} for examples)}.

To see that 
{our distance} 
is a metric, {note that} non-negativity, symmetry, and {the} triangle inequality follow directly from the symmetric difference.  For the identity condition, note that by the symmetric difference,  $d(N,N') = 0$  if and only if 
$\cD(N) = \cD(N')$ and 
$\cR(N) = \cR(N')$.
{Since, by} Theorem~\ref{the:encoding}, a network {$M\in\cA(X)$ is encoded 
by $\cR(M)\cup \cD(M)$}
{it follows that}  $d(N,N') = 0$  if and only if $N\simeq N'$.  
{We leave as an open question: what is the diameter of $\cA(X)$ under the duet-triplet distance?}

\section{Conclusion}

Arboreal networks are multi-rooted structures whose underlying graph is a tree.
We give an elegant characterization of stack-free arboreal networks in terms of the rooted 2-sets and 3-sets that they induce.  This characterization allows for {an} algorithm for reconstructing such networks with running time linear in the number of inputted duets and triplets, as well as an intuitive metric to compare the networks. We leave as a {further} open question:  can we bound the number of the triplet/duets needed to reconstruct a stack-free arboreal network up to isomorphism? 

\smallskip
\smallskip
\noindent{\bf Acknowledgements:}
The authors thank {the referees for their helpful comments. They also thank the} Institute for Computational and Experimental Research in Mathematics (ICERM) in Providence, RI (Fall 2024),
the Center for Interdisciplinary Research (ZiF), University Bielefeld, Germany (Summer 2025),
and
the Banff International Research Station (BIRS) 
(August 2025) for hosting them.

\smallskip
\noindent
{\bf Data Availability:}  No data was used.

\smallskip
\noindent
{\bf CRediT Statement:}  The authors both contributed equally to the conceptualization, the methodology, and the writing.

\small
\bibliographystyle{acm}

\bibliography{references}

\end{document}